\documentclass[10pt,twocolumn,twoside]{IEEEtran}

\usepackage{amsmath}
\usepackage{bbold}
\usepackage{enumerate}
\usepackage{cite}
\usepackage{color,xcolor}
\usepackage{amssymb}
\usepackage{graphicx}
\usepackage{algorithm}

\usepackage{algpseudocode}
\usepackage{bbm}
\usepackage{graphics}
\usepackage{subfigure}
\usepackage{epsfig}
\usepackage{epstopdf}
\usepackage{makecell}
\usepackage[utf8]{inputenc}
\DeclareMathOperator*{\minimize}{minimize}

\def\dref#1{(\ref{#1})}

\begin{document}
\newtheorem{lemma}{Lemma}
\newtheorem{assumption}{Assumption}
\newtheorem{theorem}{Theorem}
\newtheorem{proposition}{Proposition}
\newtheorem{definition}{Definition}
\newtheorem{corollary}{Corollary}
\newtheorem{remark}{Remark}
\newtheorem{problem}{Problem}

\newenvironment{proof}{\hspace{0ex}\textsc{Proof}.\hspace{1ex}}{\hfill$\Box$\newline}

\title{Data-Driven Moving Horizon Estimators for Linear Systems with Sample Complexity Analysis}


\author{Peihu Duan, Jiabao He, Yuezu Lv, Guanghui Wen

\thanks{The work was supported by the National Natural Science Foundation of China through Grant Nos. 62088101, 62325304, 62273045, 62073079, U2341213 and U22B2046, and the Beijing Nova Program through Grant No. 20230484481. }

\thanks{P. Duan and Y. Lv are with State Key Laboratory of CNS/ATM, Beijing Institute of Technology, Beijing 100081, China. E-mails: duanpeihu@bit.edu.cn(P. Duan), yzlv@bit.edu.cn(Y. Lv). }

\thanks{J. He is with School of Electrical Engineering and Computer Science, KTH Royal Institute of Technology, Stockholm, Sweden. E-mail: jiabaoh@kth.se.}

\thanks{G. Wen is with the Department of Systems Science, School of Mathematics, Southeast University, Nanjing 211189, China. E-mail: wenguanghui@gmail.com.}
}

\maketitle
\begin{abstract}
This paper investigates the state estimation problem for { linear systems subject to Gaussian noise, where the model parameters are unknown}. By formulating and solving an optimization problem that incorporates both offline and online system data, a novel data-driven moving horizon estimator (DDMHE) is designed. {We prove that the expected 2-norm of the estimation error of the proposed DDMHE is ultimately bounded.
Further, we establish an explicit relationship between the system noise covariances and the estimation error of the proposed DDMHE. Moreover, through a sample complexity analysis, we show how the length of the offline data affects the estimation error of the proposed DDMHE. We also quantify the performance gap between the proposed DDMHE using noisy data and the traditional moving horizon estimator with known system matrices. Finally, the theoretical results are validated through numerical simulations.}
\end{abstract}


\IEEEpeerreviewmaketitle

\section{Introduction} \label{s1}
State estimation is a fundamental technology in control systems that {reconstructs the full system state from measured outputs}. Depending on designing criteria and application scenarios, various state estimation methods have been reported in the literature \cite{allgower1999nonlinear,liu2022remote,liu2020moving}. Among these, the moving horizon estimator (MHE) stands out as an efficient method for systems with disturbances, nonlinear dynamics, and constraints \cite{rao2001constrained,ji2015robust,muller2017nonlinear,sui2010linear}. The standard MHE reconstructs the system state by solving an optimization problem, based on a \textit{priori} mathematical system model and a sequence of measurements in a moving time window. {Related MHE results have also considered more complex settings, including unknown inputs under dynamic quantization effects \cite{zou2020movingnew} and finite-horizon estimation under binary encoding schemes \cite{li2023finite}.} On the other hand, the model may be unavailable in some practical implementations, making the traditional MHEs difficult to apply. In such cases, developing {a data-driven MHE without requiring explicit system models is essential}. Hence, this paper focuses on the design of data-driven MHEs for systems with unknown dynamics models by leveraging previous system trajectories.

Based on whether a mathematical model of an unknown system is pre-identified for the estimator design, the data-driven state estimation (DDSE) methods reported in the literature can be broadly classified into two diagrams: \textit{indirect DDSE} \cite{tsiamis2020sample,ljung1998system,Guy2022} and \textit{direct DDSE} \cite{liu2023learning,alanwar2022data,mishra2024data,shafieezadeh2018wasserstein,fiedler2021relationship}. Specifically, \textit{indirect DDSE} starts by constructing an approximate model of an observed plant using previous input-state-output trajectories of this plant, which then serves as the basis for designing the state estimator. The modeling phase in \textit{indirect DDSE} often employs techniques such as  subspace identification \cite{tsiamis2020sample} and neural networks  \cite{Guy2022}. On the other hand, \textit{direct DDSE} bypasses the need for explicit system identification by directly developing estimators from previous system trajectories. A significant theoretical foundation for \textit{direct DDSE} is Willems' fundamental lemma \cite[Theorem 1]{willems2005note}, which offers a sufficient condition under which an input-output trajectory of a linear system can be represented by another input-output trajectory. 

Considering the significant advantages of MHEs in handling disturbances and constraints \cite{alessandri2011moving}, there has been an increasing focus on integrating the aforementioned data-driven techniques from both \textit{indirect} and \textit{direct DDSE} into MHEs for systems with unknown models \cite{lowenstein2023physics,wang2023neural,Wolff2024,liu2023data,kuhl2011real,alessandri2012min,muntwiler2022learning}. To be more specific, a class of MHEs that incorporates neural networks to model unknown system dynamics \cite{lowenstein2023physics} or cost functions \cite{wang2023neural}, referred to as neural MHEs, have emerged as an important indirect method for simultaneous parameter and state estimation. Meanwhile, Wolff et al. \cite{Wolff2024} developed a direct MHE framework for unknown linear systems with bounded measurement noise and provided a comprehensive robust stability analysis. An alternative approach to directly designing MHEs for unknown systems is to formulate the problem as a multi-variable optimization problem, which can take the form of a bilinear optimization problem \cite{kuhl2011real}, a min-max optimization problem \cite{alessandri2012min}, or a differentiable convex optimization problem \cite{muntwiler2022learning}.


The aforementioned methods usually require prior input-state-output trajectories for the design of data-driven state estimators \cite{Wolff2024,liu2023data,mishra2024data,liu2023learning,alanwar2022data}. {Without any prior state information, input-output data alone can identify the system state only up to an unknown similarity transformation, since the same input-output trajectory may correspond to multiple state trajectories \cite{tsiamis2020sample}. In addition, in many practical scenarios, prior states are collected at a lower frequency than input–output data since states often represent more complex physical quantities that require specialized or time-consuming sensing.} For example, {in jacketed continuous stirred-tank reactors (CSTRs), both the reactant concentration and the reactor temperature are system states. However, temperature can be measured online at a second-level frequency, whereas concentration measurements typically rely on chemical analysis or soft sensing and are therefore available at a much lower sampling frequency (e.g., minute-level to hourly-level) \cite{bequette2003process}.} In such cases, the available prior state information is limited to a lower-frequency state trajectory, and how to leverage a prior input-output trajectory and a lower-frequency sampled state trajectory to design MHEs for unknown systems is an important yet unresolved issue.
{On the other hand, the system may be influenced by multi-source unknown disturbances \cite{liu2013moving}, which affect the system evolution. Together with measurement noise, these effects result in noisy pre-collected data, potentially degrading the performance of data-driven estimators \cite{yin2021maximum}.} To address this issue, {Lyapunov stability analysis is commonly employed to derive a linear matrix inequality-based condition that guarantees stability of estimators with respect to deterministic disturbances \cite{liu2023learning},} while finite sample analysis is conducted to determine the sample number required to achieve desired estimation accuracy against stochastic \mbox{noise \cite{tsiamis2020sample}}.
{Although some efforts have been made in noise analysis and robust estimator design, the explicit relationship between noise statistics (e.g., noise covariance), estimation error bounds, and the required sample size is not well characterized. Therefore, it is beneficial to establish such a quantitative relationship for data-driven estimators.}


Motivated by the above observation, this paper investigates the problem of learning an MHE for a linear system affected by both process and measurement noise, where the system matrices of the corresponding state-space model are unknown. The available prior information consists of a sampled input-output trajectory and a lower-frequency sampled state trajectory. The goal is to design an MHE capable of estimating an online state trajectory of the system based on the corresponding real-time input-output data and the pre-collected system trajectory. This paper {formulates} the MHE design problem into an integrated optimization problem. According to the solution to this optimization problem, a new data-driven MHE (DDMHE) is proposed (\textbf{Algorithm \ref{algorithm1}}). In comparison to existing studies, this paper has three key advantages as follows.

\begin{enumerate}
  \item {This paper designs a novel DDMHE framework that enables state estimation using prior input–output data together with sparsely sampled state data, thereby accommodating different sampling rates between prior state and input–output data, such as in CSTRs \cite{magni2001stabilizing}.}

  \item {This paper considers both process and measurement noise in the offline and online system data, and establishes an explicit analytical relationship between the system noise covariances and the estimation error of the proposed DDMHE (\textbf{Theorem \ref{thm1}}).} 

  \item {This paper analytically reveals how the performance gap between the proposed DDMHE and the traditional MHE with known system matrices decreases as the number of offline data samples increases (\textbf{Theorem~\ref{thm2}}). This result characterizes the sample complexity of the proposed DDMHE, i.e., the minimum amount of offline samples required to achieve a given state estimation accuracy}.
\end{enumerate}


%

{\it Notation:} {Let $\mathbb{N}_{>0}$ denote the set of positive integers}, and $\mathbb{N}$ denote the set of nonnegative integers. Let $\otimes$ denote the Kronecker product. Let $a_{[t_1,t_2]} \triangleq [a_{t_1}^T, \ldots, a_{t_2}^T]^T$ denote a {column vector} of the signal $a$ during the time interval $[t_1,t_2]$. Let $I$ denote the identity matrix of an appropriate dimension. Let $0$ denote the zero matrix of an appropriate dimension. For any positive definite matrix $S$, let $\lambda_{\textup{min}(S)}$ denote the minimum eigenvalue of $S$. For any matrix $S$, let $S(p_1:p_2;q_1:q_2)$ denote the block matrix in $S$ with elements $S_{mn}$, $ p_1 \leq m \leq p_2$, $q_1 \leq n \leq q_2$, and let $S^{\dagger}$ denote the right/left inverse. Let $\mathcal{N}(\theta$, $\Xi)$ denote the Gaussian distribution with mean $\theta \in \mathbb{R}^n$ and covariance $\Xi \in \mathbb{R}^{n \times n}$. Let $\mathcal{U}_{(-\delta/2, \delta/2]}$ denote the continuous uniform distribution in the interval $(-\delta/2, \delta/2]$. For a vector $\eta$ and a symmetric matrix $S$ with appropriate dimensions, let $\|\eta\|_S$ denote $\sqrt{\eta^T S \eta}$.

\section{Problem Formulation}\label{s2}

\subsection{System Description} \label{s2.1}
{This paper studies a class of linear time-invariant systems}, whose dynamics are described by
\begin{align} \label{equ:systemmodel}
\begin{split}
x_{k+1} & = A x_k + B u_k + \omega_k, \\
y_k &  = C x_k + \nu_k, \quad k \in \mathbb{N},
\end{split}
\end{align}
where $k$ is the time index; $x_{k} \in \mathbb{R}^{n}$ and $u_{k} \in \mathbb{R}^{m}$ denote the system state and input of the system, respectively;  $y_{k} \in \mathbb{R}^{p}$ represents the measurement; $\omega_k \in \mathbb{R}^{n}$ and $\nu_k \in \mathbb{R}^{p}$ represent the system process and measurement noise, respectively; and $A \in \mathbb{R}^{n \times n}$, $B \in \mathbb{R}^{n \times m}$, and $C \in \mathbb{R}^{p \times n}$ are unknown system matrices. In this model, we assume $\omega_k$ and $\nu_k$ satisfy Gaussian distribution, i.e., $\omega_k \sim \mathcal{N}(0$, $\sigma_{\omega}^2 I_n)$ and $\nu_k \sim \mathcal{N}(0$, $\sigma_{\nu}^2 I_p)$ with $\sigma_{\omega} \in \mathbb{R}_{>0}$ and $\sigma_{\nu} \in \mathbb{R}_{>0}$, and the initial state $x_0$ satisfies $x_0 \sim \mathcal{N}(\bar{x}_0$, $\sigma_{0}^2 I_n)$ with $\bar{x}_0 \in \mathbb{R}^{n}$ and $\sigma_{0} \in \mathbb{R}_{>0} $. Suppose that $x_0$, $\omega_k$, and $\nu_k$, $\forall k \in \mathbb{N}$, are mutually uncorrelated.  
During the time interval $[k-L,k]$, where $L$ is a positive integer, it follows from \eqref{equ:systemmodel} that {the output sequence over the horizon} satisfies
\begin{align}  \label{equ:yinterval}
y_{[k-L,k]}  & = G x_{k-L}  +    H  u_{[k-L,k-1]}    +   F \omega_{[k-L,k-1]}  +  \nu_{[k-L,k]}   \notag \\
  &   \triangleq \Psi_L (x_{k-L}, u_{[k-L,k-1]}, \omega_{[k-L,k-1]}, \nu_{[k-L,k]} ),
\end{align}
where
\begin{align}  \label{equ:GFH}
  G =  \left [       
  \begin{array}{c}   
    C  \\
    CA \\
    \vdots  \\
    CA^{L}  \\
  \end{array} \right ],
  F =  \left [                
  \begin{array}{ccccc}   
    0 & 0 & 0 & 0  \\
    C & 0 & 0 & 0 \\
    CA & C & 0 & 0 \\
    \vdots  &  & \ddots  &    \vdots \\
    CA^{L-1} & CA^{L-2} & \cdots & C  \\
  \end{array} \right ]
\end{align}
and $H = F (I_L \otimes B)$. {Here, $\Psi_L(\cdot)$ is introduced as a compact notation to stack the output variables over the horizon into a single vector, facilitating the formulation of the moving horizon estimation problem.
}

\begin{assumption} \label{asm:observable}
{In system \dref{equ:systemmodel}, the pair $(C$, $A)$ is observable.}
\end{assumption}

{
When Assumption~\ref{asm:observable} holds and $L \ge n$, the observability matrix $G$ defined in \eqref{equ:yinterval} is of full column rank, which is a requirement in moving horizon estimation to ensure that the state estimate is uniquely determined \cite{liu2013moving}.}

\begin{assumption} \label{assumptionstatebound}
For system \dref{equ:systemmodel}, there exist two positive scalars $\pi_1$ and $\pi_2$ such that $\mathbb{E} \{  x_k^T x_k  \} \leq \pi_1  $ and  $  \mathbb{E} \{  u_k^T u_k  \} \leq \pi_2  $, $\forall k \ge 0$.
\end{assumption}

\begin{remark}
{In the data-driven setting, noisy data may lead to inaccuracies in the learned system representation. In this case, Assumption \ref{assumptionstatebound} is used to ensure stability of the estimation process and bounded estimation error, and is adopted in robust filtering for guaranteeing stability and performance \cite{xie1994robust,shaked2002robust,garcia2012robust}.}
\end{remark}

 \begin{figure}[t]
  \centering
   {\includegraphics[scale=0.45]{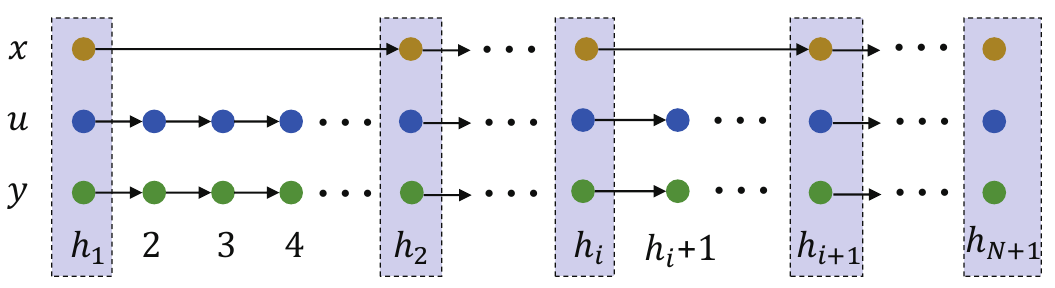}}
  \caption{{The pre-collected state-input-output system trajectory, where   yellow, blue, and green solid dots denote the state, input, and output, respectively. The states are sampled at time instants $h_1$, $h_2$, $\ldots$, $h_{N+1}$, and the inputs and outputs are sampled at every time instant from $h_1$ to $h_{N+1}$.}}
  \label{fig:priorT}
\end{figure}

\subsection{Data Collection} \label{s2.3}
{In this paper}, we assume that we have collected a prior  input-state-output trajectory of system \dref{equ:systemmodel}, which is shown in Fig. \ref{fig:priorT}. We consider the scenario where the state sampling frequency is much lower than the input-output sampling frequency. Hence, it is reasonable to assume that there exists a positive integer $L$ such that $L_i > L$ holds for all $ i \in \mathcal{V} \triangleq \{1,\ldots,N\}$, where $L_i = h_{i+1} - h_i$. {If this condition is not satisfied for some $i$, the $h_{i+1}$-th state sample can be skipped and a further one can be used until the condition holds.} Next, {we partition the pre-collected trajectory into $N$ segments, where the $i$-th segment corresponds to the time interval $[h_i,\, h_i+L]$.}
{We define the stacked input and output data within each segment as}
\begin{align} \label{equ:Tsegment}
\begin{split}
u_{[0,L-1],i}
& \triangleq
[u_{h_i}^\top, \ldots, u_{h_i+L-1}^\top]^\top, \\
y_{[0,L],i}
& \triangleq
[y_{h_i}^\top, \ldots, y_{h_i+L}^\top]^\top.
\end{split}
\end{align}
{As discussed in Section \ref{s1}, input–output data alone are insufficient to uniquely determine the state trajectory
when the system model is unknown. In this paper, the offline dataset includes a low-frequency state trajectory. Since these state measurements may be subject to noise, each segment is locally re-indexed from time $0$, and its initial state is denoted by $x_{0,i} \triangleq x_{h_i}$, which corresponds to the true system state at the sampling instant $h_i$ and is generated by the system dynamics. However, $x_{0,i}$ is not directly accessible. Instead, only a noisy measurement $\bar{x}_{0,i}$ is available, satisfying
\begin{align}\label{equ:chi}
\bar{x}_{0,i} = x_{0,i}  + \chi_{i},   \ i \in \mathcal{V},
\end{align}
where $\chi_i \sim \mathcal{N}(0$, $\sigma_{\chi}^2 I_n)$ denotes the measurement noise with $\sigma_{\chi} \in \mathbb{R}_{>0}$.} Similarly, the process and measurement noise of the $i$-th segment are denoted by $ \omega_{[0,L-1],i}$ and $ \nu_{[0,L],i}$ respectively.
The stacked prior states, inputs, and outputs of all the segments  are represented by
\begin{align}
\bar{x}_0^{\textup{p}}  \! = \! \left[ \!{ \begin{array}{c}
   {\bar{x}_{0,1}} \\
   {\vdots} \\
   {\bar{x}_{0,N}} \end{array}} \!\right], \quad
u^{\textup{p}}  \! = \! \left[ \!{ \begin{array}{c}
   {u_{[0,L-1],1}} \\
   {\vdots} \\
   {u_{[0,L-1],N}} \end{array}} \!\right], \quad
y^{\textup{p}}  \! = \! \left[ \!{ \begin{array}{c}
   {y_{[0,L],1}} \\
   {\vdots} \\
   {y_{[0,L],N}} \end{array}} \!\right],
    \notag
\end{align}
where the superscript `p' represents that data are pre-collected.
Altogether, the pre-collected system data are summarized as
\begin{align} \label{equ:mathcalmp}
\mathcal{M}^{\textup{p}} & \triangleq \big \{\bar{x}_0^{\textup{p}}, \ u^{\textup{p}}, \ y^{\textup{p}}  \big \}.
\end{align}
Let
\begin{align} \label{equ:Up}
\bar{X}_0 \!=\! [\bar{x}_{0,1},\ldots,\bar{x}_{0,N}],   U^{\textup{p}} \!=\! [u_{[0,L-1],1},\ldots,u_{[0,L-1],N}].
\end{align}
Two assumptions about the pre-collected data are made.

\begin{assumption} \label{assumptionx0u}
 $ \textup{rank}  \bigg  [
 \begin{array}{c}
   \bar{X}_0^{\textup{p}} \\
   U^{\textup{p}} \\
 \end{array} \bigg ]   = n + Lm $.
\end{assumption}

\begin{remark}
{Assumption \ref{assumptionx0u} is a persistent excitation condition on the collected data. It ensures that the input sequence is sufficiently informative so that the underlying system dynamics can be properly captured from the data.}
\end{remark}

%

\begin{assumption} \label{assumptionL}
 $ L \ge \max \{n, \ m, \ p \}  $.
\end{assumption}

\begin{remark}
In fact, only $L \ge n$ is required for the estimator design.
The condition $L \ge \max\{m, p\}$ is adopted to {simplify the expression of the sample complexity bound in Section IV. Here, the sample complexity bound refers to a lower bound on the number of offline data samples required to guarantee a prescribed estimation accuracy with high probability.}
\end{remark}


\subsection{Problem Statement} \label{s2.4}
This paper studies the state estimation problem of an online trajectory of system \dref{equ:systemmodel}, where the state, input, output, and noise sequences are denoted by
\begin{align} \label{equ:onlinetraj}
x_{[0,t]}, \ u_{[0,t-1]}, \ y_{[0,t]}, \ \omega_{[0,t-1]}, \ \nu_{[0,t]}, \ {t \in \mathbb{N}_{>0}},
\end{align}
respectively. All the available information regarding this online trajectory is represented by
\begin{align} \label{equ:mathcalmo}
\mathcal{M}^{\textup{o}} & \triangleq \big \{ u_{[0,t-1]} , \ y_{[0,t]} \big \}.
\end{align}
This paper {aims to design} a DDMHE to estimate the state of the above online trajectory, based on the offline and online data $\mathcal{M}^{\textup{p}}$ and $\mathcal{M}^{\textup{o}}$. Specifically, let $\bar{x}_{t-L}$ and $\hat{x}_{t-L|t}$ denote the prior and posterior estimates of the state $x_{t-L}$, respectively. {The objective is to design $\bar{x}_{t-L}$ and $\hat{x}_{t-L|t}$ using a moving horizon estimation framework}. Without loss of generality, we assume $t \ge L$ for concise presentation.

\noindent
\textbf{Problem 1}: Considering an online trajectory denoted by \dref{equ:onlinetraj} of system \dref{equ:systemmodel} with unknown $A$, $B$, and $C$, given $\mathcal{M}^{\textup{p}}$,  $\mathcal{M}^{\textup{o}}$, and {$\bar{x}_{t-L}$ at the time step $t$}, derive $\hat{x}_{t-L|t}$ by solving the following minimization problem
\begin{align}  \label{equ:optimization1}
  &   \minimize_{\xi_{\textup{on}|t}, \ \xi_{\textup{off}|t}} \quad J_{\textup{on}}+ J_{\textup{off}}
  \end{align}
subject to the following constraints:
 \begin{align}  \label{equ:const1}
  & \quad y_{[t-L,t]} \! = \!   \Psi_L (\hat{x}_{t-L|t}, u_{[t-L,t-1]},  \hat{\omega}_{[t-L,t-1]}, \hat{\nu}_{[t-L,t]} ),  \\
  & \quad  \label{equ:const2} y_{[0,L],i}    \! = \! \Psi_L (\hat{x}_{0,i}, u_{[0,L-1],i},  \hat{\omega}_{[0,L-1],i}, \hat{\nu}_{[0,L],i}  ), 
\end{align}
where $\Psi_L $ is given in \dref{equ:yinterval}; the optimization variables $\xi_{\textup{on}|t}$ and $\xi_{\textup{off}|t}$ are defined as
$$\xi_{\textup{on}|t} = \{\hat{x}_{t-L|t},\hat{\omega}_{[t-L,t-1]}, \hat{\nu}_{[t-L,t]}\},$$
$$\xi_{\textup{off}|t} = \{ \hat{x}_{0,i}, \hat{\omega}_{[0,L-1],i},  \hat{\nu}_{[0,L],i}, i \in \mathcal{V} \},$$
respectively; $L$ is the length of the sliding window; and the cost functions $J_{\textup{on}}$ and $J_{\textup{off}}$ are  defined as
 \begin{align}  
J_{\textup{on}} = \alpha \|\hat{x}_{t-L|t} - \bar{x}_{t-L}\|^2 \! + \! \sum_{h=t-L}^{t-1} \! \| \hat{\omega}_{h} \|^2_{\sigma_{\omega}^{-2}} \! + \! \sum_{h=t-L}^{t} \! \| \hat{\nu}_{h} \|^2_{\sigma_{\nu}^{-2}}, \notag
\end{align}
and
 \begin{align}  
J_{\textup{off}} \! =\! \sum_{i=1}^{N} \Big ( \|\hat{x}_{0,i} \! - \! \bar{x}_{0,i} \|^2_{\sigma_{\chi}^{-2}} \!+\! \sum_{h=0}^{L-1} \| \hat{\omega}_{h,i}  \|^2_{\sigma_{\omega}^{-2}} \!+\! \sum_{h=0}^{L} \| \hat{\nu}_{h,i}  \|^2_{\sigma_{\nu}^{-2}}  \Big ),  \notag
\end{align}
with $\alpha$ being a positive constant.
{Further, design the prior state estimate $\bar{x}_{t-L+1}$ based on $\hat{x}_{t-L|t}$ for state estimation at next time step.}

\begin{remark}
{The variables $\xi_{\mathrm{on}|t}$ and $\xi_{\mathrm{off}|t}$ are introduced to represent the unknown state and noise sequences in the online and offline trajectories, respectively. They serve as optimization variables that account for the mismatch between the predicted quantities and the measured data over the estimation horizon. The objective function combines online and offline terms to improve estimation performance while handling unknown system dynamics.
If the system matrices $A$, $B$, and $C$ are known, $J_{\textup{off}}$ and \dref{equ:const2} can be omitted such that Problem 1 will be converted into the traditional moving horizon estimation problem  \cite{alessandri2003receding,zou2020moving}.}
\end{remark}


Then, two specific subproblems arise as follows.

1)  Find a solution to the optimization problem \dref{equ:optimization1} with unknown system matrices $A$, $B$, and $C$, and design a DDMHE based on the solution.

2) {Analyze the estimation performance of the proposed DDMHE by characterizing its sample complexity, namely, the minimum amount of offline data required to achieve a given state estimation accuracy.}

\section{DDMHE Design} \label{s3}

In this section, a solution to the optimization problem \dref{equ:optimization1} is derived. Based on this solution, a recursive DDMHE is proposed.
To move on, we define several notations regarding the pre-collected offline data $\mathcal{M}^{\textup{p}}$ and the variable $\xi_{\textup{off}|t}$ by
\begin{align}
Y^{\textup{p}} & = [y_{[0,L],1},\ldots,y_{[0,L],N}],  &   & \hat{V} = [\hat{\nu}_{[0,L],1},\ldots,\hat{\nu}_{[0,L],N}], \notag \\
 \hat{\Omega} & =  [\hat{\omega}_{[0,L-1],1}, \ldots, \hat{\omega}_{[0,L-1],N}] , &  & \hat{X}_0   = [\hat{x}_{0,1},\ldots,\hat{x}_{0,N}]. \notag
\end{align}
It follows from \dref{equ:const2} that
\begin{align}  \label{equ:augdata}
Y^{\textup{p}}  & = G \hat{X}_0 + H   U^{\textup{p}}    +   F \hat{\Omega}  +  \hat{V},
\end{align}
where $U^{\textup{p}}$ is defined in \dref{equ:Up}. When $[\hat{X}_0^T,(U^{\textup{p}})^T,\hat{\Omega}^T]^T$ {is full row rank}, by substituting the above expression into \dref{equ:const1}, the optimization problem \dref{equ:optimization1} is equivalent to
\begin{align}
  &  \quad  \quad  \quad  \quad  \minimize_{\xi_{\textup{on}|t}, \ \xi_{\textup{off}|t}} \quad J_{\textup{on}} + J_{\textup{off}} \notag \\
\textup{s.t.}
  & \  y_{[t-L,t]}  =  (Y^{\textup{p}} -  \hat{V} ) \left[  { \begin{array}{c}
   {\hat{X}_0} \\
   { U^{\textup{p}}} \\
   {\hat{\Omega} } \end{array}} \!\right]^{\dagger} \left[  { \begin{array}{c}
   {\hat{x}_{t-L|t}} \\
   {  u_{[t-L,t-1]}} \\
   {\hat{\omega}_{[t-L,t-1]} } \end{array}} \!\right] + \hat{\nu}_{[t-L,t]}.  \notag
\end{align}
{The estimate $\hat{x}_{t-L|t}$ is one of the optimization variables and can be obtained from the numerical solution of the problem, which can be solved using standard methods such as the alternating direction method of multipliers  \cite{neal2011distributed}.} However, {due to the nonlinear constraint and the full row rank condition on $ [\hat{X}_0^T,(U^p)^T,\hat{\Omega}^T]^T$, the above problem is nonconvex. This implies that multiple local minima may exist and the obtained solution may depend on the initialization and the numerical algorithm, so global optimality cannot be guaranteed.} For these reasons, { instead of solving the above nonconvex problem directly at each time step, we derive a recursive suboptimal solution that is more suitable for real-time implementation while still guaranteeing the estimation performance.}

\begin{algorithm}[t]
  \caption{DDMHE.} \label{algorithm1}
  \hspace*{0.02in}

{\bf Input:} $\mathcal{M}^{\textup{p}}$, $\mathcal{M}^{\textup{o}}$, $\alpha$, $\sigma_{\omega}$, and $\sigma_{\nu}$;

{\bf Output:} $\hat{x}_{t-L|t}$, $t \in \mathbb{N}_{>0}$, $t \ge L$;

\begin{algorithmic}[1]

\State  compute matrices $G_*$ and $H_*$:
  \begin{align} \label{equ:GHstar}
      [G_*,  \ H_*] = Y^{\textup{p}}  \left[  { \begin{array}{c}
   {\bar{X}_0^{\textup{p}}} \\
   { U^{\textup{p}}} \end{array}}  \right]^{\dagger};
   \end{align}

 \State  construct matrix $F_*$:
    \begin{flalign}
   F_* \longleftarrow 0_{(L+1)p \times Ln};  \notag
   \end{flalign}
    {\bf for $h =1, \ldots, L$ do}
     \begin{equation} \label{equ:Fstar}
    \begin{split}
     & \ F_* (hp+1:Lp+p;hn-n+1:hn)   \\
     =  & \   G_*  (1:Lp-hp+p;1:n)  ;
     \end{split}
   \end{equation}
    {\bf end for}
\State  {\bf for $t =L, L+1, \ldots$ do}
\begin{align}
  \begin{split} \label{equ:hatx}
  &  \hat{x}_{t-L|t} \! = \!  \Lambda_*  [\alpha_1 \bar{x}_{t-L} \!+\!    \Gamma_* (z_{[t-L,t]}  \!  -   \! H_* \! u_{[t-L,t-1]} )],   \\
  & {\bar{x}_{t-L+1}   \! = \!  \Phi_{*, 1}^{\dagger} [ \Phi_{*, 2}, \  \Phi_{*, 3}] [\hat{x}_{t-L|t}^T, \  u_{t-L}^T ]^T},
  \end{split} \\
  & \textup{where } \notag \\
  &  \Gamma_* \! = \!  \alpha_2 G_*^T  (\alpha_2 I \!+\! F_* F_*^T )^{-1} ,   \Lambda_* \! =\! (\alpha_1 I \! + \! \Gamma_* G_*)^{-1} ,   \notag \\
   & \alpha_1 \! = \!\alpha \sigma_{\nu}^{2}, \ \alpha_2 \! = \!  \sigma_{\nu}^{2} / \sigma_{\omega}^{2},  \   \Phi_{*, 1}   = G_*  (1:Lp;1:n), \notag \\
    &  \Phi_{*, 2}    = G_*  (p+1:Lp+p;1:n),   \notag \\
    &  \Phi_{*, 3}    = H_*  (p+1:Lp+p;1:m); \notag
 \end{align}
  {\bf end for}
\end{algorithmic}
\end{algorithm}

Specifically, {an approximate method is proposed, in which the original optimization problem \dref{equ:optimization1} is decomposed into a two-step optimization procedure as follows.}
\begin{align}   \label{equ:optimization2}
\begin{split}
  \minimize_{ \xi_{\textup{on}|t}|\xi_{\textup{off}|t}^{*} }      \ \quad J_{\textup{on}},   \quad
\textup{s.t.}
        \   \dref{equ:const1} \  \text{and} \ \dref{equ:const2},
 \end{split}
\end{align}
where $\xi_{\textup{on}|t}|\xi_{\textup{off}|t}^{*}$ denotes the optimization variable $\xi_{\textup{on}|t}$ given a specific value of $\xi_{\textup{off}|t}$, denoted by $\xi_{\textup{off}|t}^{*}$, which is the optimal solution to
\begin{align}  \label{equ:optimization3}
\begin{split}
  \minimize_{\xi_{\textup{off}|t} }    \quad J_{\textup{off}},   \quad
\textup{s.t.}
     \ \dref{equ:const2}.
\end{split}
\end{align}
%
%
%
Subsequently, by solving \dref{equ:optimization2} and \dref{equ:optimization3}, the explicit expression of $\hat{x}_{t-L|t}$ is specified as \mbox{Algorithm \ref{algorithm1}}, which is called the DDMHE. The derivation of Algorithm \ref{algorithm1} is provided in Appendix \ref{algorithm1derivation}.

\begin{remark}
The proposed DDMHE is completely established on noisy data $\mathcal{M}^{\textup{p}}$ and $\mathcal{M}^{\textup{o}}$, without using any knowledge of system matrices $A$, $B$, and $C$. 
{In addition, the proposed MHE formulation can be extended to incorporate state and input constraints. In particular, using the same construction as in the unconstrained case, the resulting quantities are incorporated into the online optimization problem together with additional state and input constraints. The resulting constrained problem can be solved numerically.}
\end{remark}

%
%

\section{DDMHE Performance Analysis} \label{s4}
In this section, we derive the finite sample complexity for learning the DDMHE parameters. Further, the boundedness of the estimation error is ensured. Moreover, the performance gap between the DDMHE and the traditional MHE using known system matrices is established.

\subsection{Sample Complexity for Learning DDMHE Parameters} \label{s4.1}
In this subsection, we provide a finite sample complexity analysis for the DDMHE parameters, {which are directly determined from data rather than obtained via system identification.} First, in the experiment of generating data, let $u_{h,i} \sim \mathcal{N}(0$, $\sigma_u^2 I_m)$ with $\sigma_u  \in \mathbb{R}_{>0} $, $h  \in \{ h_i,h_i+1,\ldots,h_i+L-1\}$,  $i \in \mathcal{V}$. {This choice is motivated by its rich excitation properties and its common use in system identification and data-driven estimation \cite[Chapter 13.3]{ljung1998system}.} Since system \dref{equ:systemmodel} is a Gaussian process, we assume that $\bar{x}_0^i$ is a zero-mean Gaussian variable and its  covariance is $\sigma_{\bar{x}_0}^2 I_n$ with  $\sigma_{\bar{x}_0} \in \mathbb{R}_{>0} $. Moreover, we assume that $u_{h,i}$ and $\bar{x}_{0,i}$, $\forall h  \in \{ h_i,h_i+1,\ldots,h_i+L-1\}$, $\forall i \in \mathcal{V} $, are  mutually uncorrelated. Then, we define
$ \Delta_{\Phi}  \! = \! \Phi_{*, 1}^{\dagger} [ \Phi_{*, 2},   \Phi_{*, 3}] -  [ A , \ B ],  \Delta_{G} \! = \! G_*-G,  \Delta_{H} \! = \! H_* - H.$
Now, we propose a result {regarding} the upper bounds of $\Delta_{\Phi}$, $\Delta_{G}$ and $\Delta_{H}$ in terms of the length $N$ as follows.

\begin{proposition} \label{proposition1}
Consider system \dref{equ:systemmodel} with pre-collected data $\mathcal{M}^{\textup{p}}$ defined in \dref{equ:mathcalmp}, and suppose Assumptions \ref{asm:observable}, \ref{assumptionx0u}, and \ref{assumptionL} hold. For any two small positive scalars $ \theta \in (0, \ 1)$ and $\varepsilon > 0$, there exists a positive scalar $N_0(\varepsilon, \theta)$ that if $N \ge N_0(\varepsilon, \theta)$, then
\begin{align} 
\begin{split}
\| \Delta_{\Phi} \|_2 \leq \varepsilon, \ \| \Delta_{G} \|_2 \leq \varepsilon, \ \| \Delta_{H} \|_2 \leq \varepsilon,  \notag
\end{split}
\end{align}
simultaneously hold with probability at least $1-\theta$. {Moreover, $\|\Delta_\Phi\|_2$, $\|\Delta_G\|_2$, and $\|\Delta_H\|_2$ decay at a rate of $\mathcal{O}(N^{-1/2})$.}
\end{proposition}

The proof of Proposition \ref{proposition1} is provided in Section \ref{proofpro1}. By combining \dref{equ:ghbound3} with \dref{equ:phibound} in the proof and utilizing the union bound, when
\begin{align} \label{equ:epsilonbound}
\varepsilon < \varepsilon_0 \triangleq \sqrt{\|\Phi_{1}\|_2^2 + \lambda_{\min}(\Phi_{1}^T \Phi_{1}) } - \|\Phi_{1}\|_2,
\end{align} with $ \Phi_{1}   \triangleq G  (1:Lp;1:n)$ with $G$ defined below \dref{equ:yinterval}, one feasible $N_0(\varepsilon, \theta)$ can be derived as
\begin{align} \label{equ:N0}
N_0(\varepsilon, \theta) \! = \! 16 L^2  \! + \! [16 L^2 \! + \! ( M_1^2 \! + \! 1 ) M_0^2 / \varepsilon^2] \textup{log}(324/\theta),
 \end{align}
where
\begin{align}
& M_0  =  48 L \sigma_{\textup{max}}^2 \sigma_{\textup{min}}^{-2} (\|F \|_2 + \|G \|_2 +1 ), \notag \\
& M_1  = ( \|A\|_2 + \|B\|_2 + 2)/\sqrt{\lambda_{\textup{min}}( \Phi_{1}^T \Phi_{1}) - \varepsilon^2 - 2 \varepsilon \| \Phi_{1} \|_2} ,  \notag \\
& \sigma_{\textup{max}} = \max \{\sigma_{\omega}, \sigma_{\nu}, \sigma_{u}, \sigma_{\bar{x}_0}, \sigma_{\chi}  \},      \notag \\
& \sigma_{\textup{min}}   = \min \{  \sigma_{u}, \sigma_{\bar{x}_0} \}.    \notag
 \end{align}

\begin{remark} \label{remark3}
{Although the bound in (20) is expressed in terms of the system matrices, this mainly serves to show explicitly how the system properties affect the sample complexity, as is common in finite-sample analysis. In practice, these quantities can be replaced by their data-driven counterparts, for example, $\|A\|_2+\|B\|_2$, $\|G\|_2$, and $\|H\|_2$ by $\|\Phi_{*,1}^{\dagger}[\Phi_{*,2},\Phi_{*,3}]\|_2+\epsilon$,  $\|G_*\|_2+\epsilon$ and $\|H_*\|_2+\epsilon$, respectively.} Moreover, Proposition \ref{proposition1} offers a finite sample complexity analysis for learning the DDMHE parameters, establishing a direct relationship between the learning error $\varepsilon$ and the sample length $N_0$. Specifically, it can be found from \dref{equ:N0} that a smaller $\varepsilon$ requires a larger $N_0$. { Similarly, a smaller confidence parameter $\theta$ leads to a larger required sample size $N_0$, which reflects the standard trade-off between confidence level and data requirement in high-probability finite-sample analysis.}
\end{remark}

Next, we generalize the result obtained in Proposition \ref{proposition1} to cases where the system noise satisfies sub-Gaussian distribution defined below. {This extension allows us to cover a broader class of disturbances beyond the Gaussian case, including bounded noise and disturbances arising from sensor saturation, quantization, or the aggregation of multiple independent noise sources \cite{vershynin2018high}.}

\begin{definition} \cite{vershynin2018high}
A random vector $X \in \mathbb{R}^d$ is sub-Gaussian if there is a positive number $\sigma_0$ such that
\begin{align}
\mathbb{E}[e^{\lambda^T (X-\mu)}] \leq e^{ \sigma_0^2 \|\lambda\|^2/2}, \notag
\end{align}
for all $\lambda \in \mathbb{R}^d$. In this case, we denote $X \sim \textup{subG}(\mu, \sigma_0)$.
\end{definition}

\begin{corollary} \label{cor0}
Considering that noise in \dref{equ:systemmodel} and \dref{equ:chi} follows sub-Gaussian distribution, i.e., $\omega_k \sim  \textup{subG}(0, \sigma_{\omega})$,  $\nu_k \sim  \textup{subG}(0, \sigma_{\nu})$, and $\chi_{i} \sim \textup{subG}(0, \sigma_{\chi})$, the result obtained in Proposition \ref{proposition1} remains valid.
\end{corollary}

Since Proposition~\ref{proposition1} relies on Lemmas~\ref{lemma1} and \ref{lemma2}, which also hold for sub-Gaussian random variables~\cite{vershynin2018high}, Corollary~\ref{cor0} follows directly. Consequently, the subsequent results remain valid for sub-Gaussian noise, covering a class of disturbances such as bounded, uniform, Laplace, and Bernoulli noise.

\subsection{Boundedness Analysis of DDMHE Estimation Error} \label{s4.2}

In this subsection, we provide a detailed analysis for the estimation error of the proposed DDMHE. Let
\begin{align}   \label{equ:definee}
  e_{t-L} \triangleq \hat{x}_{t-L|t} - x_{t-L},
\end{align}
denote the state estimation errors of the online trajectory \dref{equ:onlinetraj} at time step $t-L$ using the proposed DDMHE.


\begin{theorem} \label{thm1}
Consider system \dref{equ:systemmodel} with pre-collected data $\mathcal{M}^{\textup{p}}$ defined in \dref{equ:mathcalmp}, and suppose that Assumptions \ref{asm:observable}, \ref{assumptionx0u},  \ref{assumptionL}, and \ref{assumptionstatebound} hold.
For any two positive scalars $ \theta \in (0, \ 1)$ and $0 < \varepsilon < \varepsilon_0 $ with $\varepsilon_0$ defined in \dref{equ:epsilonbound}, if $N \ge N_0(\varepsilon, \theta)$ with  $N_0(\varepsilon, \theta)$ defined in \dref{equ:N0} and  $0<c_1<1$ with $c_1$ defined below, then
\begin{align}   \label{equ:infinitye}
 \lim_{t \rightarrow \infty} \mathbb{E} \{ \| e_{t-L} \|_2 \} \leq  \frac{c_2}{1 - c_1} ,
\end{align}
holds with probability at least $1- \theta$, where
\begin{align}
c_1 = & \ \alpha \sigma_{\nu}^{2} \|    \Lambda_*  \Phi_{*, 1}^{\dagger}  \Phi_{*, 2} \|_2, \notag \\
  c_2 = & \ \Big [\alpha \sigma_{\nu}^{2} \sigma_{\omega} \sqrt{n}    + \sigma_{\max} \sqrt{ (\sqrt{L} \varepsilon \! + \!  \|F_* \|_2)^2  L n \! + \! 2  (L \! + \! 1 )p }   \notag  \\
   &  + \varepsilon  ( \alpha \sigma_{\nu}^{2} + \| \Gamma_* \|_2) ( \sqrt{\pi_1} + \sqrt{\pi_2}) \Big ] \|  \Lambda_* \|_2 . \notag
\end{align}
\end{theorem}

The proof of Theorem \ref{thm1} is given in Section \ref{proofthm1}. Note that the condition $0<c_1<1$ is used in \mbox{Theorem \ref{thm1}}, which must hold when selecting the parameter $\alpha$ in $J_{\textup{on}} $ to be sufficiently small. For instance, when $\alpha$ is selected satisfying
$$\alpha^{-1} > \frac{(\| \Phi_{*, 1}^{\dagger}   \Phi_{*, 2}  \|_2 - 1)  \sigma_{\nu}^{2} }{\lambda_{\textup{min}}(\Gamma_* G_*)}, $$
we can directly derive that $0<c_1<1$ must hold. Particularly, when $\| \Phi_{*, 1}^{\dagger}   \Phi_{*, 2}  \|_2 \leq 1$, it is sufficient to set $\alpha$ as any positive scalar.


{Theorem~\ref{thm1} is similar in spirit to \cite[Theorem~1]{alessandri2003receding}, which establishes bounded estimation error for model-based MHE with known system matrices. By contrast, Theorem~~\ref{thm1} shows that such a guarantee can still be obtained for the proposed DDMHE in the present noisy data-driven setting with unknown system matrices.} In addition, \mbox{Theorem \ref{thm1}} indicates that the expected value of the 2-norm of the estimation error using the proposed DDMHE is ultimately bounded. Moreover, it explicitly reveals how the characteristics of data noise (e.g., $ \sigma_{\omega}, \sigma_{\nu}, \sigma_{\chi}$), system dimensions (e.g., $n,p$), and system matrices  (e.g., $ \Gamma_*, \Lambda_*$) affect the estimation error.

\begin{corollary} \label{cor1}
Suppose that Assumptions \ref{asm:observable}, \ref{assumptionx0u}, and \ref{assumptionL} hold. For $0<c_1<1$ and {$\sigma_{\omega} = \sigma_{\nu} = \sigma_{\chi}  = 0$}, the proposed DDMHE guarantees that $e_{t-L}$ is asymptotically stable in the mean sense, i.e.,
\begin{align}
 \lim_{t \rightarrow \infty} \| e_{t-L} \|_2 = 0. \notag
 \end{align}
\end{corollary}

Corollary \ref{cor1} can be proved by setting the values of $ \sigma_{\omega}$, $\sigma_{\nu}$ and  $\sigma_{\chi}$ as zero in Theorem \ref{thm1}. {Corollary \ref{cor1} indicates that the proposed DDMHE can accurately achieves perfect asymptotic estimation of the system state in the absence of system noise.}

\subsection{Finite-Sample Performance Gap of DDMHE} \label{s4.3}

In this subsection, we derive a sample complexity bound for learning the designed DDMHE, which evaluates the estimation performance gap between the DDMHE and the MHE with known system matrices. To move on, the traditional MHE based on known system matrices, referred to as model-based MHE (MBMHE), is given as follows \cite{zou2020moving,alessandri2003receding}:
\begin{align} \label{equ:MHE}
\begin{split}
   &  \hat{{x}}_{t-L|t}^{\textup{m}}= \Lambda [\alpha_1 \bar{x}_{t-L}^{\textup{m}} +  \Gamma  (z_{[t-L,t]}  \!  -  H u_{[t-L,t-1]} )],   \\
   & \bar{x}_{t-L+1}^{\textup{m}}  = A \hat{x}_{t-L|t}^{\textup{m}} + B u_{t-L},
  \end{split}
 \end{align}
where $ \hat{{x}}_{t-L|t}^{\textup{m}}$ is the estimate of $x_{t-L|t}$, and
 $$   \Gamma =  \alpha_2 G^T  (\alpha_2 I \!+\! F  F^T )^{-1} , \quad \Lambda = (\alpha_1 I + \Gamma G)^{-1} ,$$
with $F$, $G$, and $H$ being {defined in \dref{equ:GFH}}, and $\alpha_1$ and $\alpha_2$ being defined in Algorithm \ref{algorithm1}. Similarly to \dref{equ:definee}, let the estimation error of the above model-based MHE be denoted by
\begin{align}   \label{equ:defineemb}
  e_{t-L}^{\textup{m}} \triangleq \hat{x}_{t-L|t}^{\textup{m}} - x_{t-L}.
\end{align}

\begin{theorem} \label{thm2}
Consider system \dref{equ:systemmodel} with pre-collected data $\mathcal{M}^{\textup{p}}$ defined in \dref{equ:mathcalmp}, and suppose that Assumptions \ref{asm:observable}, \ref{assumptionx0u},  \ref{assumptionL}, and \ref{assumptionstatebound} hold.
For any two positive scalars $ \theta \in (0, \ 1)$ and $0 < \varepsilon < \varepsilon_0 $ with $\varepsilon_0$ defined in \dref{equ:epsilonbound}, if $N \ge N_0(\varepsilon, \theta)$ with  $N_0(\varepsilon, \theta)$ defined in \dref{equ:N0}, $0<c_1<1$ and $0<c_1^{\textup{m}}<1$, then
\begin{align}   \label{equ:infinityecom}
 \lim_{t \rightarrow \infty} \mathbb{E} \{ \| e_{t-L} - e_{t-L}^{\textup{m}}  \|_2 \} \leq  \mathcal{O}(N^{-1/2} ),
\end{align}
holds with probability at least $1-  \theta$, where $c_1$ is defined below \dref{equ:infinitye} and  $ c_1^{\textup{m}} \triangleq \alpha \sigma_{\nu}^{2} \|  \Lambda A \|_2$.
\end{theorem}

The proof of Theorem \ref{thm2} is provided in Section \ref{proofthm2}. {Theorem \ref{thm2} shows that the estimation performance of the proposed DDMHE converges to that of the traditional MBMHE at a rate of $\mathcal{O} (N^{-1/2})$. This rate characterizes the sample complexity of the proposed DDMHE, i.e., the minimum amount of offline data required to achieve a given estimation accuracy. Moreover, the results of this paper are developed for linear time-invariant systems. For certain time-varying systems, such as linear switched systems composed of linear time-invariant subsystems, the proposed framework can be applied to each subsystem. Extensions to general linear time-varying or nonlinear systems are left for future research.}

\section{Simulation} \label{s5}
{In this section, the effectiveness of the proposed DDMHE is illustrated by a numerical simulation of a series elastic actuator (SEA)-driven robotic system, which commonly applies to human-robot interaction \cite{li2016adaptive}. A general dynamic model of a SEA-driven robot is described by
\begin{align}
& \left [
  \begin{array}{c}
    \ddot{q} \\
    \dot{q} \\
    \ddot{\theta} \\
    \dot{\theta}
  \end{array}
\right]
=
A
\left [
  \begin{array}{c}
    \dot{q} \\
    q \\
    \dot{\theta} \\
    \theta
  \end{array}
\right]
+
\left [
  \begin{array}{cc}
    M_1^{-1} & 0 \\
    0 & 0 \\
    0 & D_1^{-1} \\
    0 & 0
  \end{array}
\right]
\left [
  \begin{array}{c}
    \tau_e \\
    \tau
  \end{array}
\right],
\notag
\end{align}
and $y = [q,\ \theta]^\top$, where $q$, $\dot{q}$, and $\ddot{q}$ denote the position, velocity, and acceleration of the robot joint, respectively, and $\theta$, $\dot{\theta}$, and $\ddot{\theta}$ denote the position, velocity, and acceleration of the actuator, respectively. Moreover,
\begin{align}
A =
\left [
  \begin{array}{cccc}
    -M_1^{-1}C_1 & -M_1^{-1}K & 0 & M_1^{-1}K \\
    1 & 0 & 0 & 0 \\
    0 & D_1^{-1}K & -D_1^{-1}D_2 & -D_1^{-1}K \\
    0 & 0 & 1 & 0
  \end{array}
\right ].
\notag
\end{align}
All parameters in the above model are defined in \cite{li2016adaptive}. For this system, the vector $[\dot{q}, q, \dot{\theta}, \theta]^\top$ is the system state, the vector $[\tau_e, \tau]^\top$ is the system input, and the vector $[q,\theta]^\top$ is the system output. Similarly to \cite{li2016adaptive}, we use $M_1 = 0.3 \,\mathrm{kg\cdot m^2}$, $C_1 = 0.1 \,\mathrm{kg\cdot m^2/s}$, $D_1 = 0.2 \,\mathrm{kg\cdot m^2}$, $D_2 = 1 \,\mathrm{kg\cdot m^2/s}$, and $K = 1 \,\mathrm{N\cdot m/rad}$ to generate the simulation data, while all these parameters are assumed to be unknown to the proposed estimator. Then, with sampling period $T_s = 0.01\,\mathrm{s}$, the corresponding discrete-time model can be written in the form of \dref{equ:systemmodel}, where
\[
x_k = [\dot{q}_k,\ q_k,\ \dot{\theta}_k,\ \theta_k]^\top,\
u_k = [\tau_{e,k},\ \tau_k]^\top,\
y_k = [q_k,\ \theta_k]^\top,
\]
and
\begin{align}
A &= \begin{bmatrix}
0.997 & -0.033 & 0 & 0.033 \\
0.010 & 1.000 & 0 & 0 \\
0 & 0.049 & 0.951 & -0.049 \\
0 & 0 & 0.010 & 1.000
\end{bmatrix}, \notag\\
B &= \begin{bmatrix}
0.033 & 0 & 0 & 0 \\
0 & 0 & 0.049 & 0
\end{bmatrix}^\top,\
C = \begin{bmatrix}
0 & 1 & 0 & 0 \\
0 & 0 & 0 & 1
\end{bmatrix}. \notag
\end{align}
In addition, the noise terms are chosen with $\sigma_\omega=\sigma_\nu=0.002$. In the following simulation, the matrices $A$ and $B$ are assumed to be unknown to the proposed estimator and are used only for data generation.

\begin{figure}[t]
  \centering
   \includegraphics[scale=0.35]{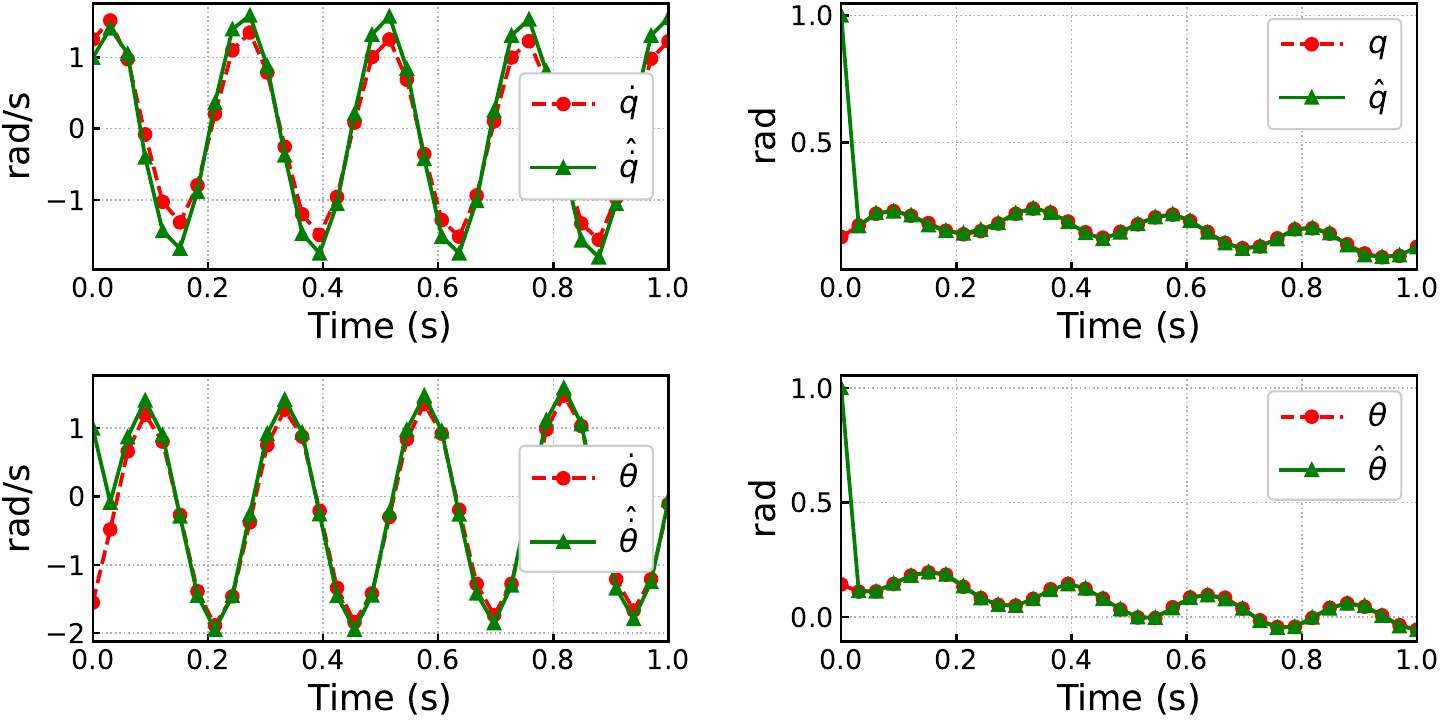}
  \caption{The true state trajectories and the corresponding estimates of the SEA-driven robotic system obtained by the proposed DDMHE.}  \label{fig:state}
  \end{figure}

For offline data collection, we generate the dataset $\mathcal{M}^{\textup{p}}$ defined in \dref{equ:mathcalmp} through numerical simulation of the above SEA-driven robotic system. Specifically, the input-output trajectory is sampled at every time instant and divided into $N$ segments, each with horizon length $L$. The corresponding input and output data are collected to form $u^{\mathrm p}$ and $y^{\mathrm p}$, while one state sample is recorded for each segment to form $\bar{x}_0^{\mathrm p}$. Unless otherwise specified, we choose $N=500$, $L=10$, $\sigma_u=10$, and $\sigma_\chi=0.01$ for data collection. Based on these offline data, we apply the proposed DDMHE, i.e., Algorithm \ref{algorithm1}, to estimate an online trajectory of the SEA-driven robotic system. The online trajectory starts from an unknown initial state and is generated under sinusoidal excitation. In the simulation, the true online state trajectory is available only for performance evaluation, while the estimator has access only to the online input-output data and the offline dataset $\mathcal{M}^{\textup{p}}$. The above simulation setting is chosen to be consistent with the assumptions used in this paper. In particular, Assumption~\ref{asm:observable} is satisfied since the pair $(C,A)$ of the SEA-driven robotic system is observable. Assumption~\ref{assumptionstatebound} is fulfilled since the online state and input trajectories remain bounded in all simulation runs under the chosen initial condition and sinusoidal excitation. Assumption~\ref{assumptionx0u} is satisfied by the offline dataset. Specifically, the input-output data are collected from $N=500$ segments using an excitation signal with amplitude level $\sigma_u=10$, and the corresponding data matrix is numerically checked to satisfy the required rank condition. Moreover, the horizon length is chosen as $L=10$, which satisfies Assumption~\ref{assumptionL}. To proceed, two types of estimation errors are defined as follows:
\begin{align}
  \textup{MSE}(j) \!=\! \frac{1}{90} \sum_{k=11}^{100} \|x_k^{(j)} - \hat{x}_k^{(j)} \|_2^2, \
  \textup{AMSE} \!=\! \frac{1}{N_{\mathrm{mc}}} \sum_{j=1}^{N_{\mathrm{mc}}} \textup{MSE}(j), \notag
\end{align}
where $j$ denotes the $j$-th Monte Carlo trial and $N_{\mathrm{mc}}=50$ is the total number of trials. Here, $\textup{MSE}(j)$ is the average mean-square estimation error of the $j$-th trial, and $\textup{AMSE}$ is the average mean-square estimation error over all trials.
 \begin{figure}[t]
  \centering
   {\includegraphics[scale=0.38]{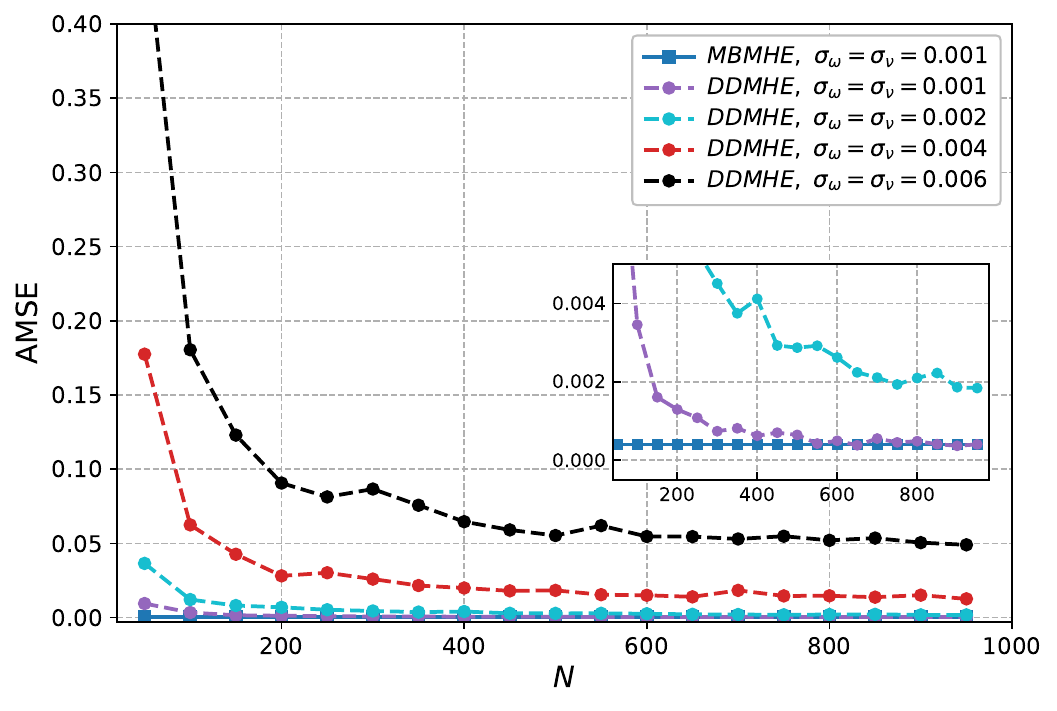}}
  \caption{The values of $\textup{AMSE}$ under different numbers of samples and magnitudes of noise using  the MBMHE and the proposed DDMHE, where the number $N$ is set as $50, 100, \ldots, 950$, respectively.  }
  \label{fig:complexity}
\end{figure}

Based on the above simulation setup, we first illustrate the basic state estimation performance of the proposed DDMHE. Fig.~\ref{fig:state} shows that the estimated positions and velocities closely track the true trajectories, indicating the effectiveness of the proposed method. Next, we present the simulation results under different noise levels and sample sizes in Fig.~\ref{fig:complexity}. It can be observed that, for a fixed sample size $N$, the estimation error increases as the noise level $(\sigma_\omega,\sigma_\nu)$ becomes larger, which is consistent with Theorem~1 where the error bound explicitly depends on the noise statistics. Moreover, as $N$ increases, the estimation performance improves and gradually approaches that of the MBMHE, validating Theorem~2 which predicts a diminishing performance gap with increasing data. In addition, Fig.~\ref{fig:complexity} also illustrates the limitation of the proposed method. Specifically, when the noise level is large and the sample size is small, the estimation error becomes significantly higher, indicating performance degradation under challenging conditions. This observation provides further insight into the reliability and practical limits of the proposed DDMHE.

In addition, we compare the proposed DDMHE with several existing estimators for systems with unknown model parameters, including MBMHE~\cite{zou2020moving,alessandri2003receding}, the data-driven Kalman filter (DDKF)~\cite{duan2023data}, the robust data-driven MHE (RDMHE)~\cite{Wolff2024}, the neural-network-based MHE (NMHE)~\cite{alessandri2011moving}, and the system-identification-based MHE (SIMHE) \cite{ljung1998system}. The results are summarized in Fig.~\ref{fig:com_001}, where it can be observed that the proposed DDMHE achieves estimation performance comparable to the MBMHE and other benchmark methods. It is worth noting that these methods either rely on a known system model or require continuously sampled state data. In terms of computational efficiency, the average computation time of the proposed DDMHE is approximately $0.72$~s per trial on a computer with an Intel processor (16 cores) and 32~GB RAM running Windows~11, which is lower than that of the RDMHE and NMHE, requiring $10.52$~s and $18.69$~min per trial, respectively. The higher cost of the RDMHE is due to solving an optimization problem at each step. For the NMHE, the reported runtime mainly arises from the offline training stage based on approximately 50{,}000 samples, while its online execution is relatively fast and comparable to that of the proposed DDMHE. Moreover, the NMHE typically relies on a known model, whereas the proposed DDMHE is developed for systems with unknown models. Overall, the simulation results support the effectiveness  of the proposed DDMHE.}

\begin{figure}[t]
  \centering
   {\includegraphics[scale=0.45]{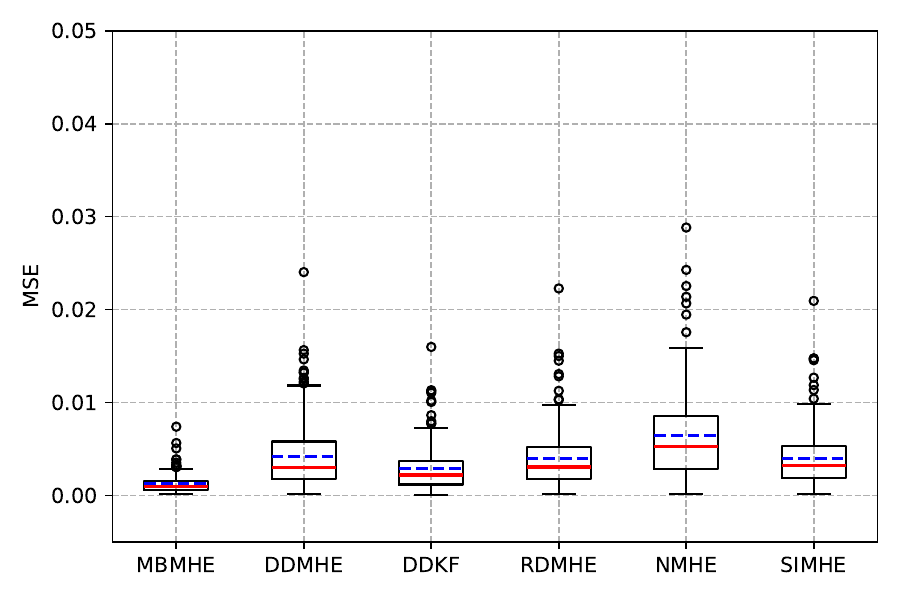}}
  \caption{The values of MSE of 200 Monte Carlo trials using different estimation methods, where the red solid line and the blue dotted line in the middle of each box denote the median and mean, respectively; the tops and bottoms of each box represent the 25th and 75th percentiles, respectively; and black circles denote outliers beyond 1.5 times the interquartile range.}
  \label{fig:com_001}
\end{figure}

\section{Conclusion} \label{s6}
In this paper, we have studied the moving horizon state estimation problem for a linear Gaussian system, where the system matrices are unknown and the measurements are collected in a binary encoding scheme. A novel DDMHE has been proposed, which depends on previous system input-output trajectories with approximate initial states. We have guaranteed that the 2-norm of the estimation error using the proposed DDMHE is ultimately bounded in probability. Further, we have compared its performance with the traditional MHE based on known system matrices. A numerical simulation has been conducted to verify the effectiveness of these results.

\section{Appendix}

\subsection{Three Basic Lemmas} \label{usefullemms}

\begin{lemma} \label{lemma1} \cite[Lemma 1]{dean2020sample}
Consider  $\Phi = [\phi_1, \ldots, \phi_N] \in \mathbb{R}^{m_1 \times N}$ and $\Psi = [\psi_1, \ldots, \psi_N] \in \mathbb{R}^{m_2 \times N}$, where $\phi_i \sim \mathcal{N}(0$, $\sigma_{\phi}^2 I_{m_1})$ and $\psi_i \sim \mathcal{N}(0$, $\sigma_{\psi}^2 I_{m_2})$, $\forall i=1,\ldots,N$, are i.i.d. random variables. For any positive scalar $\theta \in (0, \ 1 )$, when $N \ge 2  (m_1 + m_2) \textup{log}(1/\theta)$,
\begin{align}
\|\Phi \Psi^T \|_2 \leq 4 \sigma_{\phi} \sigma_{\psi} \sqrt{N (m_1 + m_2) \textup{log}(9/\theta)}, \notag
\end{align}
holds with probability at least $1-\theta$.
\end{lemma}

\begin{lemma}  \label{lemma2} \cite[Lemma 2]{dean2020sample}
Consider   $\Phi = [\phi_1, \ldots, \phi_N] \in \mathbb{R}^{m \times N}$, where $\phi_i \sim \mathcal{N}(0$, $\sigma_{\phi}^2 I_{m})$, $\forall i=1,\ldots,N$, are i.i.d. random variables. For any positive scalar $\theta \in (0, \ 1 )$,
\begin{align}
\lambda^{\frac{1}{2}}_{\textup{min}}(\Phi \Phi^T) \ge  \sigma_{\phi} \Big ( \sqrt{N} - \sqrt{n} - \sqrt{2\textup{log}(1/\theta)} \Big), \notag
\end{align}
holds with probability at least $1-\theta$.
\end{lemma}
%
%
%
%

\subsection{Derivation of Algorithm \ref{algorithm1}} \label{algorithm1derivation}

First of all, the optimal solution to the optimization problem \dref{equ:optimization3} is derived. Let $\xi_{\textup{off}|t}^* $ denote $ \{ \hat{x}_{0,i}  = \bar{x}_{0,i} $, $\hat{\omega}_{[0,L-1],i}  =0$, $\hat{\nu}_{[0,L],i} =0$, $\forall i \in \mathcal{V} \}$. It can be found that $\xi_{\textup{off}|t}^*$ is a feasible solution to \dref{equ:optimization3} with $J_{\textup{off}} (\xi_{\textup{off}|t}^*) = 0$. On the other hand, note that $J_{\textup{off}} (\xi_{\textup{off}|t}) \ge 0$ for all feasible solutions of $\xi_{\textup{off}|t}$. Hence, the above $\xi_{\textup{off}|t}^*$ is the optimal solution to \dref{equ:optimization3}.

Next, the optimal solution to the optimization problem \dref{equ:optimization2} is derived.  By substituting the above optimal solution of $\xi_{\textup{off}|t}$ to \dref{equ:optimization3} into the constraint \dref{equ:const2}, equivalently, the constraint \dref{equ:augdata}, we have
\begin{align}  
Y^{\textup{p}}  & = G_{*} \bar{X}_0^{\textup{p}} + H_{*}   U^{\textup{p}}, \notag
\end{align}
where $G_{*}$ and $H_{*}$ denote the estimates of $G$ and $H$ defined in \dref{equ:yinterval}, respectively. Moreover, $G_{*}$ and $H_{*}$ have the same dimensions and structures as $G$ and $H$, respectively, except that they are constructed by the estimates of the real system matrices $A$, $B$, and $C$, which are denoted by $A_{*}$, $B_{*}$, and $C_{*}$. Besides, $F_*$ is similarly defined. If Assumption \ref{assumptionx0u} holds, we directly have \dref{equ:GHstar}. Further, considering the structures of $G_{*}$, $H_{*}$, and $F_*$ with respect to $A_{*}$, $B_{*}$, and $C_{*}$ stated above, we can derive \dref{equ:Fstar} and
 \begin{align}
  & \Phi_{*, 1} [ A_{*}, B_{*}] = [ \Phi_{*, 2},  \Phi_{*, 3}], \ C_{*}  = G_* (1:p;1:n), \notag
 \end{align}
i.e.,
  \begin{align} \label{equ:ABCstar}
  & [ A_{*}, \ B_{*}] =  \Phi_{*, 1}^{\dagger} [ \Phi_{*, 2}, \  \Phi_{*, 3}],   \ C_{*}  \! = \! G_* (1:p;1:n),
 \end{align}
when $ \Phi_{*, 1} $ has full column rank, where $ \Phi_{*, 1} $, $ \Phi_{*, 2} $, and $ \Phi_{*, 3} $ are defined in Algorithm \ref{algorithm1}. Based on the above estimates of matrices $G$, $H$, and $F$ derived from the constraint \dref{equ:const2}, the optimization problem \dref{equ:optimization2} can be converted into
\begin{align}   \label{equ:newoptimization2}
\begin{split}
  \minimize_{ \xi_{\textup{on}|t} }  &   \qquad \qquad \qquad \qquad J_{\textup{on}}   \\
\textup{s.t.}
 \quad  & \quad y_{[t-L,t]}  =   G_* \hat{x}_{t-L|t} + H_* u_{[t-L,t-1]}  \\
   & \qquad \qquad \ \ \ + F_*  \hat{\omega}_{[t-L,t-1]} + \hat{\nu}_{[t-L,t]}.
 \end{split}
\end{align}
To solve the above optimization problem, we substitute its constraint into the expression of $J_{\textup{on}} $ to remove the variable $\hat{\nu}_h$, $h=t-L, \ldots, t$. Then, we take the partial derivative of $J_{\textup{on}}$ with respect to $ \hat{x}_{t-L|t}$ and $ \hat{\omega}_{[t-L,t-1]}$, respectively, and set them zero, i.e.,
\begin{align}
 \frac{\partial  J_{\textup{on}}}{\partial  \hat{x}_{t-L|t}} = 0, \qquad
  \frac{\partial  J_{\textup{on}}}{\partial  \hat{\omega}_{[t-L,t-1]}} = 0.\notag
\end{align}
After some computation, the above equations give rise to
\begin{align}
  \hat{x}_{t-L|t}  \! =  &  (\alpha_1 I + G_*^T G_* )^{-1} [\alpha_1 \bar{x}_{t-L} +  G_*^T (z_{[t-L,t]}  \notag \\
   &   - H_* u_{[t-L,t-1]} - F_*  \hat{\omega}_{[t-L,t-1]} )], \notag
\end{align}
and
\begin{align}
\hat{\omega}_{[t-L,t-1]}  \! =  &    (\alpha_2 I + F_*^T F_* )^{-1} [ F_*^T (z_{[t-L,t]}  - H_* u_{[t-L,t-1]} \notag \\
    &  -  G_* \hat{x}_{t-L|t} )], \notag
\end{align}
respectively. Substituting the expression of $\hat{\omega}_{[t-L,t-1]}$ into the one of $  \hat{x}_{t-L|t}  $ yields \dref{equ:hatx}. Besides, we design the $\bar{x}_{t-L+1} $ for step $t+1$ based on $\hat{\omega}_{[t-L,t-1]}$. Specifically, by following the dynamics of the plant \dref{equ:systemmodel} and utilizing the estimates of $A$ and $B$ in \dref{equ:ABCstar}, $\bar{x}_{t-L+1} $ is designed as
\begin{align}
 \bar{x}_{t-L+1}  =  A_* \hat{x}_{t-L|t} + B_*  u_{t-L} , \notag
\end{align}
equivalently, the second equation in \dref{equ:hatx}. Overall, the proposed data-driven moving horizon estimator by solving  \dref{equ:optimization2} and \dref{equ:optimization3} is summarized in Algorithm \ref{algorithm1}.

\subsection{Proof of Proposition \ref{proposition1}} \label{proofpro1}

First of all, the upper bounds of $\Delta_{G}$ and $\Delta_{H}$  are derived. It follows from \dref{equ:yinterval}  that
\begin{align}
 [G,  \ H] = \big ( Y^{\textup{p}} + G \chi^{\textup{p}}  - F \Omega^{\textup{p}}  -  V^{\textup{p}} \big ) \left[  { \begin{array}{c}
   {\bar{X}_0^{\textup{p}}} \\
   { U^{\textup{p}}} \end{array}}  \right]^{\dagger},
\end{align}
where $Y^{\textup{p}}$, $\bar{X}_0^{\textup{p}}$, and $ U^{\textup{p}}$ are given above \dref{equ:augdata} and \mbox{Assumption \ref{assumptionx0u}}, respectively, and
$V^{\textup{p}} = [ {\nu}^{1}_{[0,L]},\ldots, {\nu}^{N}_{[0,L]}]$,  $ \Omega^{\textup{p}}   =  [ {\omega}^{1}_{[0,L-1]}$, $\ldots$,  ${\omega}^{N}_{[0,L-1]}]$, $ \chi^{\textup{p}}  = [\chi^{1},\ldots,\chi^{N}]$.
Next, by referring to \dref{equ:GHstar}, when \mbox{Assumption \ref{assumptionx0u}} holds, we have
\begin{align} \label{equ:ghbound}
  & [\Delta_G,  \ \Delta_H] = \big (    F \Omega^{\textup{p}}  +  V^{\textup{p}} - G \chi^{\textup{p}}  \big ) \left[  { \begin{array}{c}
   {\bar{X}_0^{\textup{p}}} \\
   { U^{\textup{p}}} \end{array}}  \right]^{\dagger}   \\
  = \ & \big (   F \Omega^{\textup{p}}  +  V^{\textup{p}}  - G \chi^{\textup{p}} \big ) \bigg[  { \! \begin{array}{c}
   {\bar{X}_0^{\textup{p}}} \\
   { U^{\textup{p}}} \end{array}}  \! \bigg]^{T} \bigg ({ \bigg [  { \! \begin{array}{c}
   {\bar{X}_0^{\textup{p}}} \\
   { U^{\textup{p}}} \end{array}}   \!\bigg] \bigg[  \! { \begin{array}{c}
   {\bar{X}_0^{\textup{p}}} \\
   { U^{\textup{p}}} \end{array}}  \! \bigg]^{T} }\bigg  )^{-1}. \notag
\end{align}
The upper bounds of all terms in \dref{equ:ghbound} are analyzed as follows. It follows from Lemmas \ref{lemma1} that
\begin{align}
 F \Omega^{\textup{p}} \bigg[  { \! \begin{array}{c}
   {\bar{X}_0^{\textup{p}}} \\
   { U^{\textup{p}}} \end{array}}  \! \bigg]^{T}  & \leq 4 \sigma_{\textup{max}}^2 \| F\|_2 \sqrt{N ( n + 2 Lm ) \textup{log}(36/\theta)}, \notag \\
 V^{\textup{p}} \bigg[  { \! \begin{array}{c}
   {\bar{X}_0^{\textup{p}}} \\
   { U^{\textup{p}}} \end{array}}  \! \bigg]^{T}&  \leq 4 \sigma_{\textup{max}}^2 \sqrt{N (Lp +p + n + Lm ) \textup{log}(36/\theta)}, \notag \\
 - G \chi^{\textup{p}}  \bigg[  { \! \begin{array}{c}
   {\bar{X}_0^{\textup{p}}} \\
   { U^{\textup{p}}} \end{array}}  \! \bigg]^{T}&  \leq 4 \sigma_{\textup{max}}^2    \| G\|_2  \sqrt{N (2 n + Lm ) \textup{log}(36/\theta)}, \notag
\end{align}
 hold with probability at least $1 - \theta/4$, when  $N \ge 2 (  n + 2Lm ) \textup{log}(4/\theta)$,  $N \ge 2 ( Lp +p + n + Lm ) \textup{log}(4/\theta)$, and $N \ge 2 ( 2 n +  Lm ) \textup{log}(4/\theta)$, respectively. Moreover, it follows from Lemma \ref{lemma2} that
 \begin{align} \label{equ:assumprove}
   \bigg [  { \! \begin{array}{c}
   {\bar{X}_0^{\textup{p}}} \\
   { U^{\textup{p}}} \end{array}}   \!\bigg] \bigg[  \! { \begin{array}{c}
   {\bar{X}_0^{\textup{p}}} \\
   { U^{\textup{p}}} \end{array}}  \! \bigg]^{T} & \ge  \sigma_{\textup{min}}^2   \Big ( \sqrt{N} - \sqrt{n + Lm} - \sqrt{2\textup{log}(5/\theta)} \Big)^2 I  \notag \\
   & \ge \frac{\sigma_{\textup{min}}^2 N}{4}  I,
\end{align}
holds with probability at least $1 - \theta/4$, when $ N \ge 8 (n+Lm) +16 \textup{log}(4/\theta)$. In this case, we have
\begin{align}
 \bigg \| \bigg ({ \bigg [  { \! \begin{array}{c}
   {\bar{X}_0^{\textup{p}}} \\
   { U^{\textup{p}}} \end{array}}   \!\bigg] \bigg[  \! { \begin{array}{c}
   {\bar{X}_0^{\textup{p}}} \\
   { U^{\textup{p}}} \end{array}}  \! \bigg]^{T} }\bigg  )^{-1} \bigg \|_2 \leq \frac{4}{\sigma_{\textup{min}}^2 N} .  \notag
\end{align}
According to the inequalities below \dref{equ:ghbound}, the union bound and Assumption \ref{assumptionL}, we can derive that
\begin{align}  
\|[\Delta_G,  \ \Delta_H ]\|_2 \leq M_0  \sqrt{\frac{\textup{log}(36/\theta)}{N}}, \notag
\end{align}
holds with probability at least $1 - \theta$, when $N \ge 16 L^2 (1 + \textup{log}(36/\theta))$,
where $ M_0 =  48 L \sigma_{\textup{max}}^2 \sigma_{\textup{min}}^{-2} (\|F \|_2 + \|G \|_2 +1 )$.
Equivalently, for a positive constant $\varepsilon$, when
 $$N \ge 16 L^2 (1 + \textup{log}(36/\theta)) + \textup{log}(36/\theta) M_0^2 / \varepsilon^2, $$
we have
\begin{align}   \label{equ:ghbound3}
\|[\Delta_G,  \ \Delta_H ]\|_2 \leq \varepsilon ,
\end{align}
holds with probability at least $1 - \theta$. Next, according to \mbox{\cite[Theorem 1]{duan2023data}}, it follows from \dref{equ:ghbound3} that
\begin{align} \label{equ:phibound}
\|\Delta_{\Phi}\|_2 \leq  \varepsilon,
\end{align}
holds with probability at least $1 - \theta$, when
 $$N \ge 16 L^2 (1 + \textup{log}(108/\theta)) + \textup{log}(108/\theta) M_1^2 M_0^2 / \varepsilon^2, $$
where
$$ M_1 = \frac{\|A\|_2 + \|B\|_2 + 2}{\sqrt{\lambda_{\textup{min}}( \Phi_{1}^T \Phi_{1}) - \varepsilon^2 - 2 \varepsilon \| \Phi_{1} \|_2}}, $$
and $\varepsilon < \varepsilon_0$ with $\varepsilon_0 \triangleq \sqrt{\|\Phi_{1}\|_2^2 + \lambda_{\min}(\Phi_{1}^T \Phi_{1}) } - \|\Phi_{1}\|_2$.
Thus, the proof of Proposition \ref{proposition1} is complete.

\subsection{Proof of Theorem \ref{thm1}} \label{proofthm1}
First, let $\bar{e}_{t-L} \triangleq \bar{x}_{t-L} - x_{t-L}$. It follows from \dref{equ:systemmodel} and \dref{equ:hatx} that
\begin{align} 
\bar{e}_{t-L} = \Phi_{*, 1}^{\dagger}  \Phi_{*, 2} e_{t-L-1} + \Delta_{\Phi} \left[ \! { \begin{array}{c}
   {x_{t-L-1}} \\
   {u_{t-L-1}}  \end{array}}  \! \right]  - \omega_{t-L-1}, \notag
\end{align}
and
\begin{align} 
e_{t-L}\!  = \!  \alpha_1 \Lambda_*  \bar{e}_{t-L} \!  - \!  \Lambda_*  \Gamma_* [\Delta_{G},   \Delta_{H}]\left[ \! { \begin{array}{c}
   {x_{t-L-1}} \\
   {u_{t-L-1}}  \end{array}}  \! \right]  \!  +  \! \Lambda_*  \bar{\omega} , \notag
\end{align}
where
\begin{align} 
\begin{split}
 &  \  \bar{\omega} =   F \omega_{[t-L,t-1]}  +  \nu_{[t-L,t]} + \epsilon_{[t-L,t]}. \notag
 \end{split}
\end{align}
According to the above equations, it can be derived that
\begin{align} 
e_{t-L}   = & \   \alpha_1 \Lambda_*  \Phi_{*, 1}^{\dagger}  \Phi_{*, 2} e_{t-L-1} -  \alpha_1 \Lambda_* \omega_{t-L-1}   +  \Lambda_*  \bar{\omega} \notag \\
   & + ( \alpha_1 \Lambda_* \Delta_{\Phi} - \Lambda_*  \Gamma_* [\Delta_{G},   \Delta_{H}] ) \left[ \! { \begin{array}{c}
   {x_{t-L-1}} \\
   {u_{t-L-1}}  \end{array}}  \! \right]  \!. \notag
\end{align}
By taking the 2-norm and expectation of each side of the above equation, we have
\begin{align}   \label{equ:e2}
&   \mathbb{E} \{ \| e_{t-L} \|_2  \}  \leq  \|  \alpha_1 \Lambda_*  \Phi_{*, 1}^{\dagger}  \Phi_{*, 2} \|_2  \mathbb{E} \{  \| e_{t-L-1} \|_2 \}  \notag \\ & +  \alpha_1 \| \Lambda_* \|_2  \mathbb{E} \{ \| \omega_{t-L-1}\|_2 \}   +  \| \Lambda_* \|_2  \mathbb{E} \{ \|  \bar{\omega} \|_2 \}     \\
   &  + \varepsilon  \|  \Lambda_* \|_2  ( \alpha_1 + \| \Gamma_* \|_2) ( \mathbb{E} \{ \|  x_{t-L-1} \|_2 \} + \mathbb{E} \{ \| u_{t-L-1} \|_2 \}), \notag
\end{align}
holds with probability at least $1- \theta$, when $N \ge N_0(\varepsilon, \theta)$ with  $N_0(\varepsilon, \theta)$ defined in \dref{equ:N0}.
Next, since the square root function is concave, it follows from Jensen's inequality that
$$ \mathbb{E} \{ \| \omega_{t-L-1}\|_2 \} \leq \sqrt{\mathbb{E} \{ \| \omega_{t-L-1}\|_2^2 \} } = \sigma_{\omega} \sqrt{n} .$$
Similarly, we have
\begin{align}
 & \mathbb{E} \{ \|  \bar{\omega} \|_2 \} \leq \sqrt{\mathbb{E} \{ \| \bar{\omega} \|_2^2 \} }  \notag \\
 = & \sqrt{ \|F^T F \|_2 \sigma_{\omega}^2 L n +  \sigma_{\nu}^2 (L +1 )p +  \delta^2 (L +1 )p/4}  \notag \\
 \leq  &  \sqrt{ \|F^T F \|_2 \sigma_{\max}^2 L n + 2 \sigma_{\max}^2 (L +1 )p }  \notag \\
  =  &  \sigma_{\max} \sqrt{ \|F^T F \|_2  L n + 2  (L +1 )p },   \notag
\end{align}
and
\begin{align}
  & \mathbb{E} \{ \|  x_{t-L-1} \|_2 \} + \mathbb{E} \{ \| u_{t-L-1} \|_2 \} \notag \\
 \leq  &  \sqrt{\mathbb{E} \{ \|  x_{t-L-1} \|_2^2 \}} +\sqrt{\mathbb{E} \{ \|  u_{t-L-1} \|_2^2 \}} = \sqrt{\pi_1} + \sqrt{\pi_2},  \notag
\end{align}
when Assumption \ref{assumptionstatebound} holds. Since $\| F   -  F_* \|_2 \leq \sqrt{L} \varepsilon $ when $\| \Delta_{G} \|_2 \leq \varepsilon$, we have
\begin{align}
  \|F^T  F\|_2  \leq & \|F\|_2^2 \leq (\sqrt{L} \varepsilon  + \|F_* \|_2)^2.   \notag
\end{align}
Now, substituting the above inequalities into \dref{equ:e2} yields
\begin{align}  
&   \mathbb{E} \{ \| e_{t-L} \|_2  \}  \leq c_1  \mathbb{E} \{  \| e_{t-L-1} \|_2 \} + c_2, \notag
\end{align}
which holds with probability at least $1- \theta$ when $N \ge N_0(\varepsilon, \theta)$, where $c_1$ and $c_2$ are defined in Theorem \ref{thm1}. Further, we have
\begin{align}
&   \mathbb{E} \{ \| e_{t-L} \|_2  \}  \leq c_1^{t-L}  \mathbb{E} \{  \| e_{0} \|_2 \} + c_2 \frac{1-c_1^{t-L}}{1-c_1}. \notag
\end{align}
Since $0<c_1<1$, when $t$ tends to infinity, we have \dref{equ:infinitye}. Thus, the proof of Theorem \ref{thm1} is complete.

\subsection{Proof of Theorem \ref{thm2}} \label{proofthm2}
First of all, it can be found from \dref{equ:definee} and  \dref{equ:defineemb}  that
$$e_{t-L} - e_{t-L}^{\textup{m}}  = \hat{x}_{t-L|t} - \hat{x}_{t-L|t}^{\textup{m}}  . $$
According to \dref{equ:hatx} and \dref{equ:MHE}, we can derive that
\begin{align}
 &    e_{t-L} - e_{t-L}^{\textup{m}} = \alpha_1 \Lambda_*  \Phi_{*, 1}^{\dagger}  \Phi_{*, 2} (e_{t-L-1} - e_{t-L-1}^{\textup{m}})   \notag \\
  &   \! + \! \alpha_1( \Lambda_*  \Phi_{*, 1}^{\dagger}  \Phi_{*, 2}  \! - \!    \Lambda A ) \hat{x}_{t-L|t}^{\textup{m}}  \! - \! (\Lambda_*   \Gamma_*  H_*  \! - \! \Lambda    \Gamma  H ) u_{[t-L,t-1]}  \notag \\
&  \!+ \! \alpha_1 ( \Lambda_* \Phi_{*, 1}^{\dagger}  \Phi_{*, 3} - \Lambda  B)  u_{t-L} + (\Lambda_*   \Gamma_* - \Lambda \Gamma ) z_{[t-L,t]} .  \notag
\end{align}
Next, let $V_{t-L} =  \mathbb{E} \{ \| e_{t-L}  - e_{t-L}^{\textup{m}} \|_2 \}$. By taking the 2-norm and expectation of each side of the above equation, we have
\begin{align} \label{equ:Vt}
V_{t-L}  \leq  c_1 V_{t-L-1} &  + c_3  ,
\end{align}
where
\begin{align}  \label{equ:c3}
\begin{split}
c_3  =  & \ \alpha_1 \| \Lambda_*  \Phi_{*, 1}^{\dagger}  \Phi_{*, 2} -   \Lambda A \|_2     \mathbb{E} \{ \| \hat{x}_{t-L|t}^{\textup{m}}  \|_2 \}  \\
&  + \| \Lambda_*   \Gamma_*  H_*  - \Lambda    \Gamma  H \|_2   \mathbb{E} \{ \|    u_{[t-L,t-1]} \|_2 \}   \\
&  +   \alpha_1 \| \Lambda_* \Phi_{*, 1}^{\dagger}  \Phi_{*, 3} - \Lambda  B\|_2    \mathbb{E} \{ \|  u_{t-L}   \|_2 \}   \\
&  + \| \Lambda_*   \Gamma_* - \Lambda \Gamma ) \|_2    \mathbb{E} \{ \| z_{[t-L,t]}   \|_2 \}.
\end{split}
\end{align}

In the following, we prove that if $N \ge N_0(\varepsilon, \theta)$, $c_3 \leq  \mathcal{O}(N^{-1/2} )$ holds with probability at least $1- \theta$. According to Proposition \ref{proposition1}, it suffices to prove that  $c_3 \leq  M_c \varepsilon $ holds when $ \| \Delta_{\Phi} \|_2 \leq \varepsilon$, $\| \Delta_{G} \|_2 \leq \varepsilon$, and $\| \Delta_{H} \|_2 \leq \varepsilon$ simultaneously hold, where $M_c$ is a positive constant. First, we consider the first term on the right side of \dref{equ:c3}. When $ \| \Delta_{\Phi} \|_2 \leq \varepsilon$, from the definition of $ \Delta_{\Phi}$,   we can directly have
\begin{align} \label{equ:deltaA}
 \| \Phi_{*, 1}^{\dagger}  \Phi_{*, 2} -    A \|_2 \leq \varepsilon.
\end{align}
Besides, when  $\| \Delta_{G} \|_2 \leq \varepsilon$  and $\| \Delta_{H} \|_2 \leq \varepsilon$ hold, from the definitions of $\Gamma_*$ and $\Gamma$ below \dref{equ:hatx} and \dref{equ:MHE}, respectively, we have
\begin{align}
 & \| \Gamma_* - \Gamma \|_2 \notag \\
 = & \alpha_2 \|  G_*^T  (\alpha_2 I + F_* F_*^T )^{-1} - G^T  (\alpha_2 I + F  F^T )^{-1}\|_2  \notag \\
 = &   \alpha_2 \|  G_*^T  (\alpha_2 I + F_* F_*^T )^{-1} - G^T  (\alpha_2 I + F_* F_*^T )^{-1}   \notag \\
  & \quad \quad  +G^T  (\alpha_2 I + F_* F_*^T )^{-1} - G^T  (\alpha_2 I + F  F^T )^{-1}\|_2  \notag \\
 \leq  &    \alpha_2 \|  G_*^T  (\alpha_2 I + F_* F_*^T )^{-1} - G^T  (\alpha_2 I + F_* F_*^T )^{-1} \|_2   \notag \\
  &   +  \alpha_2 \|  G^T  (\alpha_2 I + F_* F_*^T )^{-1} - G^T  (\alpha_2 I + F  F^T )^{-1}\|_2 . \notag
\end{align}
Noting that
\begin{align}
& \|  G_*^T  (\alpha_2 I + F_* F_*^T )^{-1} - G^T  (\alpha_2 I + F_* F_*^T )^{-1} \|_2   \notag \\
\leq  &  \|  \Delta_{G} \|_2  \| (\alpha_2 I + F_* F_*^T )^{-1} \|_2
\leq   \| (\alpha_2 I + F_* F_*^T )^{-1} \|_2   \varepsilon , \notag
\end{align}
and
\begin{align}
& \|  G^T  (\alpha_2 I + F_* F_*^T )^{-1} - G^T  (\alpha_2 I + F  F^T )^{-1}\|_2   \notag \\
= & \|  G^T   (\alpha_2 I + F_* F_*^T )^{-1} ( F  F^T - F_* F_*^T)  (\alpha_2 I + F  F^T )^{-1}\|_2   \notag \\
 \leq & \|  G^T   (\alpha_2 I \! + \! F_* F_*^T )^{-1} \|_2  \|   (\alpha_2 I \! + \! F  F^T )^{-1}\|_2  \|   F  F^T \! - \! F_* F_*^T \|_2  \notag \\
= &  \|  G^T   (\alpha_2 I \! + \! F_* F_*^T )^{-1} \|_2  \|   (\alpha_2 I \! + \! F  F^T )^{-1}\|_2 \notag \\ & \times \| (  F   -  F_* )  F^T +  F_* ( F - F_*)^T \|_2  \notag \\
\leq  &   \|  G^T   (\alpha_2 I \! + \! F_* F_*^T )^{-1} \|_2  \|   (\alpha_2 I \! + \! F  F^T )^{-1}\|_2 \notag \\ & \times (\| F \|_2 + \| F_* \|_2 ) \sqrt{L} \varepsilon,  \notag
\end{align}
where the conclusion that $\| F   -  F_* \|_2 \leq \sqrt{L} \varepsilon $ when   $\| \Delta_{G} \|_2 \leq \varepsilon$ is used in the last inequality, it can be derived that
\begin{align} 
\| \Gamma_* - \Gamma \|_2 \leq M_{\Gamma} \varepsilon, \notag
\end{align}
with $M_{\Gamma} = \!    \alpha_2 \|  G^T   (\alpha_2 I \! + \! F_* F_*^T )^{-1} \|_2  \|   (\alpha_2 I \! + \! F  F^T )^{-1}\|_2 \times (\| F \|_2 + \| F_* \|_2 ) \sqrt{L} + \alpha_2 \| (\alpha_2 I \! + \!  F_* F_*^T )^{-1} \|_2 $. Similarly, we can prove that there exists a positive constant $M_{\Lambda}$ such that
\begin{align} \label{equ:deltaLambda}
\| \Lambda_* - \Lambda \|_2 \leq M_{\Lambda} \varepsilon.
\end{align}
By combining \dref{equ:deltaA} and \dref{equ:deltaLambda}, we have
\begin{align}
 & \| \Lambda_*  \Phi_{*, 1}^{\dagger}  \Phi_{*, 2} -   \Lambda A \|_2 \notag \\
=  & \| \Lambda_*  \Phi_{*, 1}^{\dagger}  \Phi_{*, 2} -   \Lambda_* A +   \Lambda_* A  -   \Lambda A \|_2 \notag \\
\leq & \| \Lambda_* \|_2 \|  \Phi_{*, 1}^{\dagger}  \Phi_{*, 2} -  A \|_2 + \|   \Lambda_*    -   \Lambda\|_2 \| A \|_2 \notag \\
\leq & ( \| \Lambda_* \|_2  +  \| A \|_2 M_{\Lambda} ) \varepsilon. \notag
\end{align}
Meanwhile, similarly to the derivation of Theorem \ref{thm1}, when $0<c_1^{\textup{m}}<1$, we have that $\mathbb{E} \{ \| \hat{x}_{t-L|t}^{\textup{m}}  \|_2 \}$ is uniformly bounded. This indicates that there exists a positive constant $M_1$ such that
\begin{align}
\alpha_1 \| \Lambda_*  \Phi_{*, 1}^{\dagger}  \Phi_{*, 2} -   \Lambda A \|_2     \mathbb{E} \{ \| \hat{x}_{t-L|t}^{\textup{m}}  \|_2 \} \leq M_1 \varepsilon.  \notag
\end{align}
Similarly, the remaining terms on the right side of \dref{equ:c3} satisfy
\begin{align}
 \| \Lambda_*   \Gamma_*  H_*  - \Lambda    \Gamma  H \|_2   \mathbb{E} \{ \|    u_{[t-L,t-1]} \|_2 \} &   \leq M_2 \varepsilon \notag  \\
  \alpha_1 \| \Lambda_* \Phi_{*, 1}^{\dagger}  \Phi_{*, 3} - \Lambda  B\|_2    \mathbb{E} \{ \|  u_{t-L}   \|_2 \} &  \leq M_3 \varepsilon  \notag  \\
   \| \Lambda_*   \Gamma_* - \Lambda \Gamma ) \|_2    \mathbb{E} \{ \| z_{[t-L,t]}   \|_2 \} &  \leq M_4 \varepsilon, \notag
\end{align}
respectively, where $M_2$,  $M_3$, and $M_4$ are positive constants. Altogether, we have  $c_3 \leq M_c \varepsilon $, where $M_c =M_1 +M_2 +M_3 +M_4 $.

Now, by referring to \dref{equ:Vt}, we can derive that
\begin{align}
V_{t-L}  \leq c_1^{t-L}  V_{t-L-1} + c_3 \frac{1-c_1^{t-L}}{1-c_1}. \notag
\end{align}
Further, according to Proposition \ref{proposition1}, we have
\begin{align}
 \lim_{t \rightarrow \infty} V_{t-L} \leq \frac{c_3}{1-c_1} \leq \mathcal{O}(N^{-1/2} ). \notag
\end{align}
Hence, the proof of Theorem \ref{thm2} is complete.

\bibliographystyle{ieeetr}
\bibliography{ref}

\end{document}